\documentclass[11pt,onecolumn]{article}
\usepackage{graphicx} 
\usepackage[T1]{fontenc}
\usepackage{lmodern}
\usepackage[bottom=2cm]{geometry}
\usepackage{amsmath}
\usepackage{amssymb}
\usepackage{amsfonts}
\usepackage{enumitem}
\usepackage{mathtools}
\usepackage{systeme}
\usepackage[english]{babel}
\usepackage{amsthm}
\usepackage{thmtools}
\usepackage{amsrefs}
\usepackage[onehalfspacing]{setspace}
\usepackage{authblk}
\usepackage{microtype}

\newcommand{\Z}{{\mathbb{Z}}}
\newcommand{\QQ}{{\mathbb{Q}}}
\newcommand{\RR}{{\mathbb{R}}}
\newcommand{\C}{{\mathbb{C}}}
\newcommand{\F}{{\mathbb{F}}}

\newcommand{\CP}{\mathbb{CP}}
\newcommand{\tr}{\operatorname{tr}}

\newtheorem{theorem}{Theorem}[section]
\newtheorem{corollary}[theorem]{Corollary}
\newtheorem{lemma}[theorem]{Lemma}
\newtheorem{prop}[theorem]{Proposition}

\theoremstyle{definition}

\newtheorem{remark}[theorem]{Remark}

\usepackage[hypertexnames=false]{hyperref}

\usepackage{cleveref}
\crefname{prop}{Proposition}{Propositions}
\Crefname{prop}{Proposition}{Propositions}

\title{Topology of the Set of Entangled States}

\author[1]{Maximilian Illmer}
\author[1]{Tim Netzer}
\author[2,3]{Michael M. Wolf}
\affil[1]{Department of Mathematics, University of Innsbruck}
\affil[2]{Department of Mathematics, Technical University of Munich}
\affil[3]{Munich Center for Quantum Science and Technology}
\date{\today}

\begin{document}
\maketitle

\begin{abstract} 
    We investigate the topology of the set $\mathsf E$ of entangled bipartite density operators acting on $\C^{n_1}\otimes\C^{n_2}$. We start by showing that $\mathsf E$ is path-connected, and even simply connected except in the two-qubit case. In this exceptional case $\mathsf E$ turns out to be homotopy equivalent to the set of maximally entangled states, which itself is homeomorphic to $\mathbb{RP}^3$. Here we also compute the complete homology of the closure and interior of $\mathsf E$.
    In all larger dimensions, we show that the homology and homotopy groups of   $\mathsf E$ vanish in degrees $1\leq k\leq 2(n_1-1)(n_2-1)-2$, and all homology groups of degree $k\geq (n_1n_2)^2-3$ also vanish. This range is controlled by the space $\mathsf W$ of entanglement witnesses, which we show is highly connected beyond two qubits and homotopy equivalent to $\mathsf E$. By computing the Euler characteristic, using a torus-action fixed point argument together with Alexander duality, we show that $\mathsf E$  nevertheless has non-trivial reduced homology over every field for all $n_1, n_2 \geq 2$. 
\end{abstract}

\section{Introduction}

\emph{Entanglement} may be the deepest departure of quantum physics from its classical origins, and it is the fuel behind many applications in quantum science. Mathematically, the set of entangled states (density operators) has a simple definition yet a complicated structure. It is defined as the complement of the compact, convex set of separable (a.k.a.\ classically correlated \cite{Werner}) states, which are themselves just convex combinations of product states. This simple geometric picture stands in stark contrast to the elusiveness of a mathematical characterization beyond the lowest dimensional cases, for which the `PPT criterion' completely settles the question \cite{hhh}. While we know that, in principle, the set of entangled states is describable by finitely many polynomial inequalities in any fixed dimension, no such finite description is known. Even worse, we know that no finite number of criteria of the PPT type will suffice \cites{Skowronek,Hamza}, and deciding whether a state is entangled, as a weak-membership problem, is NP-hard \cites{Gurvits,Gharibian}.

While the geometry of the set of quantum states has been studied and used in quantum information theory intensively \cite{geometry},  comparatively little attention has been devoted to the topological properties of the set of entangled states. While geometry describes local structure, distances, and convexity, topology captures global features that remain invariant under continuous deformations, providing complementary information about the organization of the quantum state space.
In this paper, we pursue this topological perspective, i.e.\ we ask the question: \emph{What is the shape of the set of entangled states?} 

A topological characterization of entangled states is of both mathematical and physical  interest. From a mathematical perspective, it provides a global description of how entangled and separable states are arranged within the space of density operators and reveals structural properties that cannot be inferred from local geometric considerations alone.
From a physical perspective, topological properties are closely related to robustness under continuous perturbations and therefore naturally arise in the study of continuous quantum evolutions, adiabatic protocols, and quantum control. The topology of the boundary separating entangled and separable states also reflects the complexity of entanglement detection.

More broadly, investigating the topology of entangled state spaces establishes new connections between quantum information theory and fields such as algebraic topology, differential topology, and real algebraic geometry. These connections offer the possibility of importing powerful mathematical tools into the study of quantum correlations and providing a new perspective on the structure of entanglement. The results presented in this work constitute a first step in this direction.

The paper is structured as follows. \Cref{sec:pure}  computes the low-degree homotopy and homology groups of the space of pure entangled states. \Cref{sec:convex} collects some immediate consequences of the convex geometry of the space of density operators. In particular, we show that the set of entangled states is path-connected, deformation retracts to its outer boundary and is the Alexander dual of the set of singular separable states.
\Cref{sec:W} shows that the space of entangled states and the space of entanglement witnesses are homotopy equivalent. Beyond the case of two qubits, the latter space is highly connected, which implies that the low degree homotopy and reduced homology groups are trivial. Similarly, the top degree homology is shown to vanish. In \Cref{sec:Euler} we show that the Euler characteristic of the set of entangled states is zero in all dimensions. This implies  nontrival homology (over any field) and we bound the corresponding degree.
\Cref{sec:22} deals with the special case of two qubits. We show that the set of entangled states is homotopy equivalent to the set of maximally entangled states, which in turn is homeomorphic to $SO(3)$. We show that the homology of the set of entangled two qubit states remains unchanged when moving to the interior, but flips to the Alexander dual when taking the closure. 

\section{Preliminaries}

For a bipartite quantum system on $\C^{n_1} \otimes \C^{n_2}$ with total dimension $N\coloneqq n_1 n_2$ and $n_1,n_2\geq 2$ we denote the set of {\it states} (density matrices) by $$\mathsf D\coloneqq\big\{\sigma\in \text{Herm}(\C^{n_1} \otimes \C^{n_2})\cong\text{Herm}(\C^{N})\mid \sigma \geqslant 0, {\rm tr}[\sigma]=1\big\}.$$ Unless explicitly stated otherwise, interiors, closures, and boundaries are taken relative to the affine real space of trace-one Hermitian matrices. Thus $\mathsf D$ is a compact convex body in this affine space, which has dimension $N^2-1$. Its interior consists of the positive definite trace-one matrices, and its boundary $\partial\mathsf D$ is homeomorphic to $\mathbb S^{N^2-2}$ by radial projection from any interior point.
The extreme points of $\mathsf D$ are exactly the density matrices of rank $1$, also called {\it pure states}, which we denote by $\mathsf P\coloneqq {\rm ext}(\mathsf D)$. 

A bipartite state $\sigma$ is called {\it separable} if it can be written as a convex combination of tensor products of states of the corresponding subsystems, i.e.:

$$\sigma=\sum_i \lambda_i\left( \sigma_i^{(1)}\otimes\sigma_i^{(2)}\right),\ \lambda_i\geq 0,\  \sum_i\lambda_i=1,\  \sigma_i^{(j)}\geqslant 0,\ {\rm tr}[\sigma_i^{(j)}]=1.$$
We denote the set of all separable states by $\mathsf S$. It is a compact convex subset of $\mathsf D$ and since its affine hull equals the entire ambient space, it has non-empty interior. Its extreme points are the separable states of rank $1$, also called {\it pure product states}, which we denote by $\mathsf{PS}\coloneqq\mathsf P\cap \mathsf S$.

The set of {\it entangled states} is defined as the complement $\mathsf E\coloneqq\mathsf D\setminus \mathsf S$ and we write $\mathsf{PE}\coloneqq\mathsf P\cap\mathsf E$ for the set of  {\it pure entangled states}.
 
As we will see, some of the topological properties will be related to parts of the boundary of the set of entangled states. We therefore define $$\partial^* \mathsf E \coloneqq \mathsf E \cap \partial \mathsf D,\qquad \partial_* \mathsf E\coloneqq \overline{\mathsf E}\cap \mathsf S=\partial \mathsf E\setminus\partial^* \mathsf E.$$ We will call $\partial^*\mathsf E$ the {\it outer boundary} and $\partial_*\mathsf E$ the {\it inner boundary} of $\mathsf E$ in $\mathsf D$. 

After choosing a real affine coordinate system on the trace-one Hermitian matrices, we can regard all the above sets of density matrices as subsets of $\mathbb{R}^{N^2-1}$. Together with the interior $\mathsf E^\circ$ and the closure $\overline{\mathsf E}$ all these sets can be defined by first-order formulas in the language of ordered fields. Consequently, they are all semialgebraic and as such locally contractible (due to  \cite{real}*{Theorem 9.3.6}). Moreover, each compact semialgebraic set is homeomorphic to a finite simplicial complex \cite{real}*{Theorem 9.2.1} and therefore has finitely generated homology groups.

This allows us, in particular, to use the following version of \emph{Alexander duality}  \cite{hatcher}*{Theorem 3.44}, which we will apply to $\Sigma=\partial\mathsf D$:

\begin{lemma}[Alexander duality]\label{lem:alexander}
    Let $\Sigma$ be a space homeomorphic to $\mathbb S^{n}$, and let $A\subset\Sigma$ be nonempty, proper, compact, and locally contractible. Then for every $i$,
    $$
    \widetilde H_i(\Sigma\setminus A;\Z)\;\cong\;\widetilde H^{\,n-i-1}(A;\Z)
    $$
    in reduced singular homology and cohomology.
\end{lemma}

Another result that enables us to derive topological properties of  complements is the following lemma \cite{helmke}*{Lemma 5.3, Remark 5.4}. It formalizes the idea that deleting a set of codimension $q$ from a space can only create new topology starting in degree $q-1$. 

\begin{lemma}\label{lem:gp}
    Let $M$ be a smooth manifold and let $A\subset M$ be a closed subset that is a finite union of smooth submanifolds of $M$, each of real codimension at least $q$, where $q\ge2$. Then the inclusion $M\setminus A\hookrightarrow M$ induces isomorphisms
    $$
    \pi_i(M\setminus A)\cong\pi_i(M),\qquad H_i(M\setminus A;\Z)\cong H_i(M;\Z)\quad\text{for all}\quad 0\le i\le q-2.
    $$
\end{lemma}


\section{Pure entangled states}\label{sec:pure}

We can identify the set of pure states $\mathsf P$ with the complex projective space $\CP^{N-1}$ by associating a line $[\psi]\in\CP^{N-1}$ with the projector $|\psi\rangle\langle \psi|/\|\psi\|^2$.
Under this identification, the pure product states $\mathsf{PS}$ become the \emph{Segre variety}, which is the image of the embedding of $\CP^{n_1-1}\times\CP^{n_2-1}$ into $\CP^{N-1}$  obtained by projectivizing the map $(\psi_1,\psi_2)\mapsto \psi_1\otimes \psi_2 $. Since $\CP^{n_1-1}\times\CP^{n_2-1}$ has complex dimension $(n_1-1)+(n_2-1)$, $\mathsf{PS}$ has real codimension $2\big((N-1)-(n_1+n_2-2)\big)=:c$. As the set $\mathsf{PE}$ of entangled pure states is the complement of $\mathsf{PS}$ in $\mathsf{P}$, we can use Lemma \ref{lem:gp} in order to identify the low-degree homotopy and homology groups of $\mathsf{PE}$ with those of $\mathsf P$:

\begin{prop}\label{prop:pe-groups}
    $\mathsf{PE}$ is path-connected and for $0\le i\le c-2$:\vspace*{-6pt}
    $$
    \pi_i(\mathsf{PE})\cong\begin{cases}0,& i=1,\\  \Z,& i=2,\\ 0,& 3\le i\le c-2.\end{cases}\qquad
    H_i(\mathsf{PE};\Z)\cong\begin{cases}\Z,& i\ \text{even},\\ 0,& i\ \text{odd}.\end{cases}
    $$ In particular, if $(n_1,n_2)\ne(2,2)$, then $\mathsf{PE}$ is simply connected.
\end{prop}

\begin{proof}
    We apply Lemma \ref{lem:gp} with $M=\mathsf P$ and $A=\mathsf {PS}$.
    The homotopy and homology of complex projective space are standard results (cf. \cite{hatcher}*{Examples 0.6 and 4.44}): $\pi_1(\CP^{N-1})=0$, $\pi_2(\CP^{N-1})\cong\Z$, $\pi_i(\CP^{N-1})=0$ for $3\le i\le 2N-2$, and $H_i(\CP^{N-1};\Z)\cong\Z$ for even $i$ with $0\le i\le 2(N-1)$ and $0$ otherwise. Path-connectedness is the case $i=0$.
\end{proof}

In the case of two qubits, $\mathsf{PE}$ is not simply connected. In fact, in this case it is homotopy equivalent to the set of maximally entangled states, which is not simply connected in any dimension: 

\begin{prop}\label{prop:maxent}
    Suppose $n_1=n_2=n$, and let $\mathsf M\subset\mathsf{PE}$ be the set of maximally entangled pure states. Then $\mathsf M$ is homeomorphic to $PU(n)=U(n)/U(1)$ and has $\pi_1(\mathsf M)\cong\Z/n \Z$. For $n=2$, $\mathsf M\simeq\mathsf{PE}$ are homotopy equivalent and $\mathsf{M}\cong SO(3)\cong\mathbb{RP}^3$ are homeomorphic, with $H_0(\mathsf M;\Z)=\Z$, $H_1(\mathsf M;\Z)=\Z/2\Z$, $H_3(\mathsf M;\Z)=\Z$, while all other integral homology groups are trivial.
\end{prop}

\begin{proof}
    If we fix an orthonormal product basis, we can define a single maximally entangled state via the equivalence class of $\Omega=\sum_{i}e_i\otimes f_i$ and obtain  $\mathsf M=\{[(U\otimes I)\Omega]:U\in U(n)\}$. The map $U\mapsto[(U\otimes I)\Omega]$ is onto $\mathsf M$, and $(U\otimes I)\Omega$ and $(U'\otimes I)\Omega$ are equivalent iff $U^{-1}U'$ is a scalar. It therefore descends to a continuous bijection $PU(n)\to\mathsf M$, which is a homeomorphism since $PU(n)$ is compact and $\mathsf M$ is Hausdorff. The topology of $PU(n)$ has been analyzed for instance in \cite{PUn}, where also the fundamental group $\pi_1(PU(n))\cong\Z/n\Z$ can be found. For $n=2$, it is well known that $PU(2)\cong SU(2)/\{\pm I\}\cong SO(3)\cong\mathbb{RP}^3$ with homology as stated \cite{hatcher}*{Example 2.42}. Finally, the homotopy equivalence $\mathsf M\simeq\mathsf{PE}$ for $n=2$ follows from the fact that $\mathsf M$ is a strong deformation retract of $\mathsf{PE}$ in the following way: Since for $n=2$ every pure state is either a product state or already has full Schmidt-rank every element of $\mathsf{PE}$ can be written as $[(X\otimes I)\Omega]$ for some invertible $X$, and since $X\mapsto(X\otimes I)\Omega$ is a vector space isomorphism, we have $\mathsf{PE}\cong GL(2,\C)/\C^\times$. Every $X\in GL(2,\C)$ has a unique polar decomposition $X=UP$. Here $U=X (X^* X)^{-1/2}$ and $P=(X^* X)^{1/2}$ depend continuously on $X$, which  allows us to define a homotopy $\tilde{H}_t(X):=U P^{1-t}$, $t\in[0,1]$. This in turn induces a well-defined homotopy $H_t([X]):=[\tilde{H}_t( X)]$ on the quotient space $GL(2,\C)/\C^\times$ since for every $\lambda\in\C^\times$: $[\tilde{H}_t(\lambda X)]=[\tilde{H}_t( X)]$. Continuity of the homotopy is guaranteed as it is a composition of continuous maps. Moreover, $H$ gives rise to a strong deformation retraction from $GL(2,\C)/\C^\times$ onto $PU(2)$ as $H_0={\rm id}$, $H_1([X])\in PU(2)$ and for every $[U]\in PU(2), t\in[0,1]: H_t([U])=[U]$. 
\end{proof}

\section{Convex-geometric reductions}\label{sec:convex}

\begin{theorem}\label{thm:path}
    The spaces $\mathsf E$, $\overline{\mathsf E}$, and $\mathsf E^\circ$ are all path-connected. So $H_0(\mathsf E; \Z)=H_0(\overline{\mathsf E}; \Z)=H_0(\mathsf E^\circ; \Z)=\Z$.
\end{theorem}

\begin{proof}
    First let $\sigma\in\mathsf E$. By separating $\sigma$ from the compact convex set $\mathsf S$, there is an affine functional (an \emph{entanglement witness}) that is strictly larger on $\sigma$ than on every point of $\mathsf S$. Write $\sigma$ as a convex combination of pure states. At least one pure state $\rho$ in this decomposition has functional value larger than the maximum over $\mathsf S$. Hence, $\rho$ is entangled and the segment from $\sigma$ to $\rho$ remains in the same open half-space and therefore lies in $\mathsf E$. In this way, we obtain a connecting path between any entangled state and $\mathsf{PE}$, and since $\mathsf{PE}$ is path-connected (\Cref{prop:pe-groups}), $\mathsf E$ is path-connected. 

    Being the closure of a path-connected set, $\overline{\mathsf E}$ is connected. Moreover, since $\overline{\mathsf E}$ is locally contractible, connectedness implies path-connectedness.

    Finally, consider two states $\sigma_0,\sigma_1\in\mathsf E^\circ$ and a path $\gamma:[0,1]\to\mathsf E$ between them.  Since $\gamma([0,1])$ is compact and disjoint from the closed set $\mathsf S$, there is an $\varepsilon>0$ such that any state within distance $\varepsilon$  from  $\gamma([0,1])$ is still entangled. Consequently, for sufficiently small $\varepsilon>0$, the path $\tilde{\gamma}(t):=\gamma(t)(1-\varepsilon t(1-t))+\varepsilon t(1-t)I/N$, which connects $\sigma_0$ and $\sigma_1$, lies in $\mathsf E^\circ$ since for any $t\in(0,1)$ the state $\tilde{\gamma}(t)$ is positive definite and therefore in the interior of $\mathsf E$.  
\end{proof}

\begin{lemma}\label{lem:projection}
    For every $\rho\in\mathsf D$ there is a unique point $\Theta(\rho)\in\mathsf S$ minimizing $\|\rho-\theta\|_2$ over $\theta\in\mathsf S$. The map $\Theta\colon\mathsf D\to\mathsf S$ is $1$-Lipschitz continuous, and for all $\theta\in\mathsf S$ one has $\tr[(\rho-\Theta(\rho))(\theta-\Theta(\rho))]\le0$.
\end{lemma}

\begin{proof}
    Existence follows from compactness of $\mathsf S$, and uniqueness from strict convexity of the squared Hilbert-Schmidt norm on the convex set $\mathsf S$. For the inequality, consider $f(t)=\|\rho-((1-t)\Theta(\rho)+t\theta)\|_2^2$ on $[0,1]$ and observe that $0\leq f'(0)=2\tr[(\rho-\Theta(\rho))(\Theta(\rho)-\theta)]$. For Lipschitz continuity, consider $\rho,\mu\in\mathsf D$ and note that
    \begin{align*}
        \|\Theta(\rho)-\Theta(\mu)\|^2
        &\leq \|\Theta(\rho)-\Theta(\mu)\|^2
        +\underbrace{\langle\rho-\Theta(\rho),\Theta(\rho)-\Theta(\mu)\rangle}_{\geq0} \\
        &\quad+\underbrace{\langle\Theta(\mu)-\mu,\Theta(\rho)-\Theta(\mu)\rangle}_{\geq 0} \\
        &= \langle\rho-\mu,\Theta(\rho)-\Theta(\mu)\rangle
        \leq \|\rho-\mu\|\cdot\|\Theta(\rho)-\Theta(\mu)\|,
    \end{align*}
    where the first inequality uses the variational inequality for the two additional terms and the last inequality is Cauchy--Schwarz.
\end{proof}

\begin{prop}
\label{prop:Bor}
    $\partial^* \mathsf E$ is a strong deformation retract of $\mathsf E$, whereas $\partial_* \mathsf E$ is a strong deformation retract of $\overline{\mathsf E}$.
\end{prop}
\begin{proof} 
    For $\sigma\in\mathsf E$ we follow the line from $I/N\in\mathsf S^\circ$ towards $\sigma$ and possibly beyond. Since $\mathsf S$ is convex, it will stay in $\mathsf E$ beyond $\sigma$, as long as it stays in $\mathsf D$. It will thus eventually intersect $\partial^*\mathsf E$ in a unique point $\Omega(\sigma),$ which gives rise to the  continuous map $$\Omega\colon \mathsf E\to\partial^*\mathsf E$$ which clearly satisfies $\Omega{\mid_{ \partial^*\mathsf E}}={\rm id}_{\partial^*\mathsf E}$. Uniqueness and continuity here both rely on the fact that the line intersects the interior. Then  $H_1(\sigma,t)\coloneqq (1-t)\sigma +t \Omega(\sigma)$ defines a strong deformation retraction from $\mathsf E$ to $\partial^*\mathsf E$.

    Let $\Theta$ be the metric projection from Lemma~\ref{lem:projection}. If $\rho\in\overline{\mathsf E}\setminus\mathsf S=\mathsf E$, then $\Theta(\rho)$ cannot lie in the interior of $\mathsf S$, since otherwise moving from $\Theta(\rho)$ slightly toward $\rho$ would remain in $\mathsf S$ and be closer to $\rho$. Moreover, the open segment from $\rho$ to $\Theta(\rho)$ is disjoint from $\mathsf S$: otherwise a point of $\mathsf S$ on that segment would be closer to $\rho$ than $\Theta(\rho)$. Hence $\Theta(\rho)$ is a limit point of $\mathsf E$ and lies in $\partial_*\mathsf E$. If $\rho\in\partial_*\mathsf E$, then $\Theta(\rho)=\rho$. Thus $H_2(\rho,t)=(1-t)\rho+t\Theta(\rho)$ is the desired deformation retraction.
\end{proof}

This can now be used to relate the homology of the set of entangled states to the cohomology of the set of singular separable states:

\begin{corollary}\label{cor:alexander-reduction}
    For every $i$, $\widetilde H_i(\mathsf E;\Z)\cong \widetilde H^{N^2-3-i}(\mathsf S\cap \partial\mathsf D;\Z)$.
\end{corollary}

\begin{proof}
    From \Cref{prop:Bor} we obtain the homotopy equivalence $\mathsf E\simeq\partial^*\mathsf E  $. Since $\mathbb S^{N^2-2}\cong\partial\mathsf D=\partial^*\mathsf E\ \dot\cup\ (\mathsf S\cap \partial\mathsf D)$ and $\mathsf S\cap \partial\mathsf D$ is a nonempty, compact, and locally contractible proper subset of $\partial\mathsf D$, we can apply Alexander duality from \Cref{lem:alexander} with $A=\mathsf S\cap \partial\mathsf D$ and obtain the claimed result.
\end{proof}


\section{Entanglement witnesses}\label{sec:W}

The aim of this section is to show that the space $\mathsf E$ is homotopy equivalent to the space of entanglement witnesses and then to infer topological information about the former from the latter.

We call a Hermitian operator $W$ on $\C^{n_1}\otimes\C^{n_2}\cong\C^N$ an \emph{entanglement witness} if $$\lambda_{\max{}} (W)=\max_{\rho\in\mathsf D} \tr[\rho W]>\max_{\sigma\in\mathsf S} \tr[\sigma W]\eqqcolon h_{\mathsf{S}}(W),$$ i.e., if the largest expectation value that is achieved by any separable state is exceeded by that of some entangled state. In other words, if the top eigenspace of $W$ contains no product vector. This property is unchanged by adding a multiple of the identity and by multiplying by a positive scalar. In order to remove these trivial degrees of freedom, we restrict to  traceless operators of unit norm. That is, if  ${\cal H}_0$ denotes the set of traceless Hermitian operators on $\C^N$ and $\mathbb{S}({\cal H}_0):=\{X\in{\cal H}_0: \|X\|_2=1\}$ its unit sphere, we define the set of entanglement witnesses $$\mathsf W\coloneqq\big\{W\in\mathbb{S}({\cal H}_0):\lambda_{\max{}}(W)>h_{\mathsf S}(W)\big\},\quad \text{and}\quad \mathsf{W}^c:=\mathbb{S}({\cal H}_0)\setminus \mathsf W.$$
So $B\in\mathsf{W}^c$ iff the top eigenspace of $B$ contains a product vector. As before, we will infer topological properties of $\mathsf W$ from its complement $\mathsf W^c$, which we will analyze first: 

\begin{lemma}\label{lem:badcodim}
    The set $\mathsf W^c$ is compact, semialgebraic, path-connected and has real codimension at least $2(n_1-1)(n_2-1)$ in $\mathbb S(\mathcal H_0)$.
\end{lemma}

\begin{proof}
    The set $\mathsf W^c$ is closed: if $B_j\to B$ and $B_j\in\mathsf W^c$, choose product unit vectors $v_j$ in the top eigenspace of $B_j$. After passing to a subsequence, the corresponding product lines converge to a product line $[v]$ (using compactness of the Segre variety), and the continuity of $\lambda_{\max}$ gives $Bv=\lambda_{\max}(B)v$. Since $\mathbb S(\mathcal H_0)$ is compact, $\mathsf W^c$ is compact. 

    To see semialgebraicity, rewrite the condition $B\in\mathsf W^c$ by introducing a product unit vector $v$ and a real number $\lambda$ such that $Bv=\lambda v$ and $\lambda I-B\geq0$. Product vectors are characterized by the vanishing of the $2\times2$ minors after reshaping $v$ as an $n_1\times n_2$ matrix, and positive semidefiniteness is semialgebraic. The claim then follows from Tarski--Seidenberg (cf.\ \cite{real}*{Chapter 1.4, Chapter 2.2}).

    For path-connectedness, let $B\in \mathsf W^c$ and let $v$ be a unit product vector in the top eigenspace of $B$,  with eigenvalue $\lambda=\lambda_{\max}(B)>0$. Since $v$ is an eigenvector, $B=\lambda\,vv^*+B^\perp$ with $B^\perp$ supported on $v^\perp$, eigenvalues  $\beta_k\le\lambda$, and $\tr[ B^\perp]=-\lambda$. Set $\mu\coloneqq-\lambda/(N-1)<0$ and, for $s\in[0,1]$,
    $$
        B_s\coloneqq\lambda\,vv^*+(1-s)B^\perp+s\mu\,(I-vv^*).
    $$
    Each $B_s$ is traceless with $B_sv=\lambda v$, and its eigenvalues on $v^\perp$ are $(1-s)\beta_k+s\mu\le\lambda$. Hence, $v$ lies in the top eigenspace of $B_s$, and $B_s\ne0$. Thus $s\mapsto B_s/\|B_s\|_2$ is a path in $\mathsf W^c$ from $B$ to
    $$
        M_v:=\frac{vv^*-I/N}{\|vv^*-I/N\|_2}=\frac{B_1}{\|B_1\|_2}.
    $$
    As $[v]\mapsto M_v$ is continuous on the path-connected variety $\mathsf{PS}\cong\CP^{n_1-1}\times\CP^{n_2-1}$, the matrices $M_v$ form a path-connected subset of $\mathsf W^c$ to which every point of $\mathsf W^c$ is joined by such a path.

    It remains to estimate the dimension of $\mathsf W^c$. Let $N=n_1n_2$ and stratify according to the dimension $r$ of the top eigenspace, $1\leq r\leq N-1$; the value $r=N$ cannot occur, since a traceless Hermitian matrix whose top eigenspace is all of $\C^N$ would be a multiple of the identity, hence $0$, contradicting $\|B\|_2=1$. First choose a product line, i.e., an element of the Segre variety $\CP^{n_1-1}\times\CP^{n_2-1}$, which is of real dimension $2(n_1+n_2-2)$. After a product line has been chosen, the $r$-dimensional planes containing it are parametrized by the Grassmannian $\operatorname{Gr}(r-1,N-1)$ and hence contribute real dimension $2(r-1)(N-r)$. Thus the possible top eigenspaces containing a product vector have real dimension at most
    $2(n_1+n_2-2)+2(r-1)(N-r).$

    For a fixed top eigenspace $K$ of dimension $r$, a Hermitian matrix with top eigenspace exactly $K$ has the form $\lambda I_K\oplus C$, with $C\in{\rm Herm}(K^\perp)$ having all eigenvalues smaller than $\lambda$. For a fixed $K$, the pair $(\lambda,C)$ has $1+(N-r)^2$ real parameters when ignoring the trace and Hilbert-Schmidt norm constraints. Since those have independent differentials, they cut out an $(N-r)^2-1$ dimensional submanifold of admissible pairs $(\lambda, C)$. 

    Since
    $\dim_\RR\mathbb S(\mathcal H_0)=N^2-2,$ the codimension of the $r$-th stratum is at least
    \begin{align*}
        &N^2-2-\bigl(2(n_1+n_2-2)+2(r-1)(N-r)+(N-r)^2-1\bigr) \\
        &\qquad=(r-1)^2+2(n_1-1)(n_2-1)\geq2(n_1-1)(n_2-1).\qedhere
    \end{align*}
\end{proof}

Now we can use these properties to prove that the reduced homology of $\mathsf W$ vanishes in low and high degrees:

\begin{prop}[Topology of the set of entanglement witnesses]\label{prop:W}
    The space $\mathsf W$ satisfies\footnote{We write $\pi_0=0$ if $\pi_0$ is a singleton, corresponding to a single path-connected component.} 
    \begin{alignat*}{2}
        &\pi_i(\mathsf W)=0\quad\text{and}\quad
        \widetilde H_i(\mathsf W;\mathbb Z)=0
        &\qquad& \text{for } 0\le i\le 2(n_1-1)(n_2-1)-2,
        \\
        &H_k(\mathsf W;\mathbb Z)=0
        &\qquad& \text{for } k\ge N^2-3.
    \end{alignat*}
\end{prop}

\begin{proof}
    By \Cref{lem:badcodim}, $\mathsf W^c$ is a compact semialgebraic subset of $\mathbb S(\mathcal H_0)$ of real codimension at least $2(n_1-1)(n_2-1)$. Every semialgebraic set is a finite union of smooth submanifolds \cite{real}*{Chapter 9.1}. As $\mathsf W^c$ is closed, \Cref{lem:gp} (with $M=\mathbb S(\mathcal H_0), A=\mathsf W^c$ and $q=2(n_1-1)(n_2-1)$) shows that the inclusion $\mathsf W=\mathbb S(\mathcal H_0)\setminus\mathsf W^c\hookrightarrow\mathbb S(\mathcal H_0)$ induces isomorphisms on $\pi_i$ and on $H_i(-;\Z)$ for $0\le i\le 2(n_1-1)(n_2-1)-2$. Since $\mathbb S(\mathcal H_0)\cong\mathbb S^{\,N^2-2}$, all $\pi_i$ and $\widetilde H_i(-;\Z)$ are trivial in this range.

    For the high degree homology, note first that for any $B\in\mathsf W^c$ we have $$\mathsf W\subseteq\underbrace{\mathbb{S}({\cal H}_0)}_{\cong\;\mathbb{S}^{N^2-2}}\setminus \{B\}\cong \RR^{N^2-2}.$$ As $\mathsf W^c$ is closed,  $\mathsf W$ is thus homeomorphic to an open subset of $\RR^{N^2-2}$ and as such has trivial homology groups for all degrees $k\geq N^2-2$. For $k=N^2-3$ we use Alexander duality from \Cref{lem:alexander} and obtain
    $$H_k(\mathsf W;\Z)\cong \widetilde H_k\big(\mathbb{S}({\cal H}_0)\setminus\mathsf W^c;\Z\big)\cong \widetilde H^{N^2-3-k}(\mathsf W^c;\Z)=\widetilde H^{0}(\mathsf W^c;\Z)=0,$$
    where the last step is due to path-connectedness of $\mathsf W^c$ following \Cref{lem:badcodim}.
\end{proof}

Now we relate the spaces $\mathsf W$ and $\mathsf E$ topologically:

\begin{theorem}
    The homotopy equivalence $\mathsf E\simeq\mathsf W$ holds.
\end{theorem}

\begin{proof} 
    We have to show that there are continuous maps $\smash{\mathsf E\xrightarrow{R}\mathsf W\xrightarrow{\Phi}\mathsf E}$ whose compositions are homotopic to the respective identity, i.e., $R\circ\Phi\simeq{\rm id}_{\mathsf W}$ and $\Phi\circ R\simeq{\rm id}_{\mathsf E}$.

    For $\rho\in\mathsf E$, define
    $$
        R(\rho)\coloneqq\frac{\rho-\Theta(\rho)}{\|\rho-\Theta(\rho)\|_2}.
    $$
    This is continuous by Lemma~\ref{lem:projection}. Moreover, the  inequality in Lemma~\ref{lem:projection} yields 
    $$\tr[R(\rho)\sigma]\leq \tr[R(\rho)\Theta(\rho)]$$ for all $\sigma\in\mathsf S,$
    while
    $$
        \tr[R(\rho)\rho]=\tr[R(\rho)\Theta(\rho)]+\|\rho-\Theta(\rho)\|_2>\tr[R(\rho)\Theta(\rho)].
    $$
    Thus $\lambda_{\max}(R(\rho))>h_{\mathsf S}(R(\rho))$, and we have a map $R:\mathsf E\rightarrow\mathsf W$ that maps every entangled state to an entanglement witness, which `witnesses' the entanglement of the given state.

    For $W\in\mathsf W$ set $\Delta(W)\coloneqq\lambda_{\max}(W)-h_{\mathsf S}(W)>0$ and $\beta(W)\coloneqq 2N/(e\Delta(W))$. Define
    $$
        \Phi(W)\coloneqq\frac{e^{\beta(W)W}}{\tr[e^{\beta(W)W}]}.
    $$
    This is continuous, since $\lambda_{\max}$ and $h_{\mathsf S}$ are continuous. Let $\lambda_1=\lambda_{\max}(W)\geq\lambda_2\geq\cdots\geq\lambda_N$ be the eigenvalues of $W$ and set $\delta_j\coloneqq\lambda_1-\lambda_j$. Then 
    \begin{equation}\label{eq:Gibbsestimate}
        \lambda_1-\tr[W\Phi(W)]
        =\frac{\sum_j \delta_j e^{-\beta(W)\delta_j}}{\sum_j e^{-\beta(W)\delta_j}}
        \leq \sum_j \delta_j e^{-\beta(W)\delta_j}
        \leq \frac{N}{e\beta(W)}=\frac{\Delta(W)}2, 
    \end{equation}
    where we have used $xe^{-\beta x}\leq1/(e\beta)$ for $x\geq0$. Hence, $$\tr[W\Phi(W)]\ge h_{\mathsf S}(W)+\Delta(W)/2>h_{\mathsf S}(W),$$ so $\Phi(W)\in\mathsf E$ is an entangled state, whose entanglement is `witnessed' by $W$.

    For $\rho\in\mathsf E$ and $W=R(\rho)$, both $\rho$ and $\Phi(W)$ are strictly separated from $\mathsf S$ by the linear functional $\tau\mapsto\tr(W\tau)$. This holds as well for the entire line segment $(1-t)\rho+t\Phi(R(\rho))$, $0\le t\le1$, which therefore lies in $\mathsf E$. This is a homotopy $\Phi\circ R\simeq{\rm id}_{\mathsf E}$.

    For the reverse composition, fix $W\in\mathsf W$ and set $\rho=\Phi(W)$. Both $W$ and $R(\rho)$ strictly separate $\rho$ from $\mathsf S$: (\ref{eq:Gibbsestimate}) implies $\tr[W(\rho-\sigma)]\ge\Delta(W)/2>0$, and the  inequality in Lemma~\ref{lem:projection} gives $\tr[R(\rho)(\rho-\sigma)]\ge\|\rho-\Theta(\rho)\|_2>0$, for all $\sigma\in\mathsf S$. Hence, the traceless matrices $\omega_t\coloneqq(1-t)W+tR(\rho)$ satisfy $\tr[\omega_t(\rho-\sigma)]>0$ for all $t\in[0,1]$ and $\sigma\in\mathsf S$. In particular, $\omega_t\ne0$ and, since $\rho$ is a state, $\lambda_{\max}(\omega_t)\ge\tr[\omega_t\rho]>h_{\mathsf S}(\omega_t)$. Consequently, $W_t(W)=\omega_t/\|\omega_t\|_2$ lies in $\mathsf W$, and $(t,W)\mapsto W_t(W)$ is continuous with $W_0={\rm id}_{\mathsf W}$ and $W_1=R\circ\Phi$. This is a homotopy $R\circ\Phi\simeq{\rm id}_{\mathsf W}$, which completes the proof of $\mathsf W\simeq \mathsf E$.
\end{proof}

We can now use this equivalence to transfer topological information about $\mathsf W$ to $\mathsf E$ (and vice versa). Proposition \ref{prop:W} gives:

\begin{corollary}[Trivial topology of $\mathsf E$ in low and high degree]
    \label{cor:EHtails}
    The space $\mathsf E$ satisfies 
    \begin{alignat*}{2}
        &\pi_i(\mathsf E)=0\quad\text{and}\quad
         \widetilde H_i(\mathsf E;\mathbb Z)=0
        &\qquad& \text{for } 0\le i\le 2(n_1-1)(n_2-1)-2,\ \text{ and}
        \\
        &H_k(\mathsf E;\mathbb Z)=0
        &\qquad& \text{for } k\ge (n_1 n_2)^2-3.
    \end{alignat*}
    In particular, $\mathsf E$ is simply connected, unless $n_1=n_2=2$. 
\end{corollary}


\section{Euler characteristic and non-trivial homology}\label{sec:Euler}

The aim of this section is to compute the \emph{Euler characteristic} of $\mathsf E$ and thereby show that $\mathsf E$ has non-trivial topology in any  dimension. This will be achieved by studying the compact set   $\mathsf X:= \mathsf S \cap \partial \mathsf D$ of singular separable states and relating it to $\mathsf E$ via Alexander duality.  In order to simplify the analysis of $\mathsf X$, we will equip the space with the action of the group $G:=(\mathbb{S}^{1})^{n_1} \times (\mathbb{S}^{1})^{n_2}$ and relate the topology of $\mathsf X$ to that of its fixed point set under this torus action. 

To every $g_i\in (\mathbb{S}^1)^{n_i}$ we associate a diagonal unitary $U_{g_i}\in U(n_i)$ by identifying each $\mathbb{S}^1$ with the complex unit circle that contains the entries of one of the diagonal elements. Mapping any $g=(g_1,g_2)\in G$  to a diagonal unitary $U_g=U_{g_1}\otimes U_{g_2}$, we define a group action $G\times \mathsf X\rightarrow \mathsf X$ by $(g,\rho)\mapsto g\cdot\rho\coloneqq U_g \rho U_g^*$. Note that the product structure and invertibility of $U_g$ guarantee that this remains in $\mathsf S$ and $\partial \mathsf D$, respectively. So $\mathsf X$ becomes a `$G$-space' and we will denote the set of fixed points as $\mathsf X^G:=\{\rho\in\mathsf X\mid \forall g\in G: g\cdot\rho=\rho\}$.

For any space $\mathsf Y$ whose integral homology groups are finitely generated and vanish above some degree, the \emph{Euler characteristic} is defined as 
\begin{equation}\label{eq:chi}
    \chi(\mathsf Y) \coloneqq \sum_i (-1)^i {\rm rank }\; H_i(\mathsf Y;\Z) =\sum_i (-1)^i \dim_\F H_i(\mathsf Y;\F),
\end{equation}
where the equality holds over any field $\F$ due to the universal coefficient theorem \cite{hatcher}*{Section 3.A}. As $\mathsf X$ is homeomorphic to a finite simplicial complex (since it is a  compact semialgebraic set \cite{real}*{Theorem 9.2.1}) it meets the assumptions,  so that $\chi(\mathsf X)$ is well-defined. Since also the integral homology groups of $\mathsf E$ are finitely generated and vanish outside finitely many degrees, we get that (\ref{eq:chi}) can be safely applied to $\mathsf E$ as well.

Following \cite{Puppe}*{Corollary 3.1.13, Corollary 3.1.14} (see also \cite{Sikora}), the topology of $\mathsf X$ is related to the one of its fixed point set $\mathsf X^G$ via:
\begin{eqnarray}
    \chi(\mathsf X)&=&\chi\big(\mathsf X^G\big) \label{eq:chiG}\\
    \sum_{i=0}^\infty \dim_\QQ H^{m+2i}(\mathsf X;\QQ) &\geq& \sum_{i=0}^\infty \dim_\QQ H^{m+2i}\big(\mathsf X^G;\QQ\big)\qquad\forall m\geq 0.\label{eq:cohG}
\end{eqnarray}
This holds under the following assumptions, which are all met in our case:
\begin{enumerate}
    \item[($i$)] $\mathsf X$ has finite-dimensional rational cohomology. This holds for any compact semialgebraic set as those are homeomorphic to a finite simplicial complex \cite{real}*{Theorem 9.2.1}.
    \item[($ii$)] $\mathsf X$ is a finite-dimensional $G$-CW complex. This holds by \cite{ParkSuh}*{Theorem 5.2}: any semialgebraic $G$-space $\mathsf X$ that is closed in the representation space (which in our case is the space of Hermitian matrices) is, up to a $G$-equivariant homeomorphism, a finite, and thus finite-dimensional, $G$-CW complex. 
    \item[($iii$)] The set $\{G_x\mid x\in \mathsf X \}$ of \emph{orbit types} of $\mathsf X$ is finite, where $G_x\coloneqq\{g\in G\mid  gx=x\}$ is the stabilizer of $x$. For a compact Lie group $G$, this holds for any semialgebraic $G$-space (\cite{ParkSuh}*{p. 377}).
\end{enumerate}

The fixed point set is the intersection of $\mathsf X$ with the commutant of $\{U_g\mid g\in G\}$. Since the commutant of $\{U_{g_i}\in U(n_i)\mid g_i\in(\mathbb{S}^{1})^{n_i}\}$ is the set of diagonal matrices and, for von Neumann algebras, the commutant of a tensor product is the tensor product of the commutants, we get that $\mathsf X^G$ consists of all diagonal, singular separable states. More specifically, expressed in computational product basis, we obtain
$$\mathsf X^G = \Big\{\rho\in\partial\mathsf D\;\big|\;\rho=\sum_{i=1}^{n_1}\sum_{j=1}^{n_2}p_{ij}|i\rangle\langle i|\otimes|j\rangle\langle j|, p_{ij}\geq 0, \sum_{ij} p_{ij}=1\Big\}.$$
Since the probability distributions form a simplex $\Delta^{N-1}$, and $\mathsf X^G$ is on the boundary, we have homeomorphisms $\mathsf X^G\cong\partial\Delta^{N-1}\cong \mathbb{S}^{N-2}$. Now we are equipped to prove the following:

\begin{theorem}[Euler characteristic and non-trivial homology]\label{thm:chi}
    The Euler characteristic of $\mathsf E$ vanishes: $\chi(\mathsf E)=0$. This implies that $\mathsf E$ has non-trivial reduced homology $\widetilde H_*(\mathsf E;\F)\neq 0$ for every field $\F$. 
\end{theorem}

\begin{proof}
    The cohomology groups of a sphere $\mathbb{S}^n$ are trivial except in degrees $i\in\{0,n\}$ where $H^i(\mathbb{S}^n;\Z)=\Z$. Consequently, $\chi(\mathsf X)=\chi\big(\mathsf X^G\big)=\chi\big(\mathbb{S}^{N-2}\big)=1+(-1)^{N}$. 
    From here, we obtain the Euler characteristic of $\mathsf E$ via the following chain of equalities, which we justify subsequently:
    \begin{eqnarray}
        \chi(\mathsf E) &=& 1+ \sum_{i=1}^{N^2-3}(-1)^i {\rm rank }\;\widetilde H_i(\mathsf E;\Z)\label{eq:chieq1}\\ \label{eq:chieq2}
        &=&1+\sum_{j=0}^{N^2-4} (-1)^{N^2-3-j} {\rm rank}\;\widetilde H^j(\mathsf X;\Z)\\ \label{eq:chieq3}
        &=& 1+ (-1)^{N^2-3} \big(\chi(\mathsf X)-1\big)\\
        &=& 0.\nonumber
    \end{eqnarray}
    Equation (\ref{eq:chieq1}) uses that ${\rm rank}\; H_0(\mathsf E;\Z)=1$ by \Cref{thm:path} and that ${\rm rank}\; H_i(\mathsf E;\Z)={\rm rank}\; \widetilde H_i(\mathsf E;\Z)$ for all $i\geq1$.
    Equation (\ref{eq:chieq2}) uses Alexander duality via \Cref{cor:alexander-reduction} and substitutes the summation index $j=N^2-3-i$.
    In order to arrive at Equation (\ref{eq:chieq3}) we will show that $${\rm rank}\; \widetilde H^{N^2-3}(\mathsf X;\Z)={\rm rank}\; \widetilde H^{N^2-2}(\mathsf X;\Z)=0.$$ Under this assumption, we can then extend the range of summation in \eqref{eq:chieq2} and obtain the Euler characteristic expression in \eqref{eq:chieq3} by exploiting that ${\rm rank}\;  H^i(\mathsf X;\Z)={\rm rank}\;  H_i(\mathsf X;\Z)$. The latter follows from the universal coefficient theorem together with finitely generated homology. Finally, Equation \eqref{eq:chieq3} yields the result $\chi(\mathsf E)=0$ after inserting $\chi(\mathsf X)=1+(-1)^N$.

    The fact that ${\rm rank}\; \widetilde H^{N^2-3}(\mathsf X;\Z)=0$ follows from \Cref{cor:alexander-reduction}, which asserts that $\widetilde H^{N^2-3}(\mathsf X;\Z)\cong \widetilde H_0(\mathsf E;\Z)$, together with $\widetilde H_0(\mathsf E;\Z) = 0$ due to path-connectedness of $\mathsf E$ by \Cref{thm:path}. 

    That ${\rm rank}\; \widetilde H^{N^2-2}(\mathsf X;\Z)=0$ follows from the fact that $\mathsf X$ is a proper closed subset of  $\partial\mathsf D\cong\mathbb{S}^{N^2-2}$, so that $H^{N^2-2}(\mathsf X;\Z)=0$ (\cite{Brasher}*{Lemma 1}; this applies to \v{C}ech cohomology, which however agrees with singular cohomology for locally contractible spaces).

    Finally, consider the definition of the Euler characteristic in \eqref{eq:chi}. By \Cref{thm:path}, $\mathsf E$ has one path-connected component and thus $\dim_\F H_0(\mathsf E;\F)=1$. So $\chi(\mathsf E)=0$ implies that at least one other homology group must be non-trivial for any field $\F$.
\end{proof}

The following bounds the range in which a non-trivial homology group can be found, and shows that there is a non-torsion integral homology in an odd degree:

\begin{prop}
    There is at least one odd $k$ in the range $2(n_1-1)(n_2-1)-1\leq k\leq N^2-N-1$ for which ${\rm rank }\;H_k(\mathsf E;\Z)\neq 0$.
\end{prop}

\begin{proof}
    Using that $\mathsf X^G\cong\mathbb{S}^{N-2}$ and thus $H^{N-2}(\mathsf X^G;\QQ)=\QQ$, we obtain from \eqref{eq:cohG} with $m=N-2$ that
    \begin{equation}
       1 \leq\sum_{i\geq 0}\dim_\QQ H^{N-2+2i}(\mathsf X;\QQ)=\sum_{i\geq 0}{\rm rank}\; H^{N-2+2i}(\mathsf X;\Z).\label{eq:1leq}
    \end{equation}
    Here we have used in addition that $\dim_\QQ H^{i}(\mathsf X;\QQ)={\rm rank}\; H^{i}(\mathsf X;\Z)$, which follows from the universal coefficient theorem since $\mathsf X$ has finitely generated homology.

    Applying Alexander duality in the form of \Cref{cor:alexander-reduction} to \eqref{eq:1leq}, we get that there must be some $k=N^2-N-1-2i$ and $i\geq 0$ for which ${\rm rank}\; H_{k}(\mathsf E;\Z)\neq 0$. The lower bound on $k$ follows from \Cref{cor:EHtails}.
\end{proof}


\section{Two Qubits}\label{sec:22}

We have seen that the case of two qubits, i.e.\ $n_1=n_2=2$, is special among the above results. This lowest dimensional non-trivial case is also  special, as the Peres--Horodecki criterion \cite{hhh} is not only necessary but indeed sufficient for entanglement:
$$\sigma\in \mathsf E\ \Leftrightarrow\ \Gamma(\sigma)\not\geqslant 0$$
where $\Gamma$ denotes the partial transpose with respect to the second subsystem (or the first, the choice is irrelevant). Note that in this case, the ambient dimension is $(2\cdot 2)^2-1=15$.

For our first definition we need a preliminary  lemma (which is well-known, see for example \cite{stv}, we include a proof for completeness).

\begin{lemma}
\label{unique}
    Let $\sigma \in \mathsf E$. Then $\Gamma(\sigma)$ has exactly one negative eigenvalue (including multiplicities).  
\end{lemma}
\begin{proof}
    Let $\lambda_1\geqslant \lambda_2\geqslant \lambda_3\geqslant \lambda_4$ be the eigenvalues of $\Gamma(\sigma)$ and let $v_1,\ldots, v_4$ be a corresponding basis of eigenvectors. By the Peres--Horodecki criterion we know that $\lambda_4<0$. Now assume $\lambda_3<0$. We have $\C^4=\C^2\otimes \C^2\cong \text{Mat}_2(\C)$, and we denote by $V_3,V_4$ the images of $v_3,v_4$ under this isomorphism. Then the homogeneous quadratic equation $\det(xV_3+yV_4)=0$ has a non-trivial solution $(a,b)$ in $\C^2$, and thus  $v:=av_3+bv_4$ is a product vector. From $$v^*\Gamma(\sigma)v=\vert a\vert^2 \lambda_3+\vert b\vert^2\lambda_4<0$$ we obtain a contradiction, since $\Gamma(\sigma)$ is still block positive, i.e.\ nonnegative as a quadratic form on product vectors.
\end{proof}

Now for  $\sigma \in \mathsf E$ let  $v$ be a normalized eigenvector to the unique negative eigenvalue of $\Gamma(\sigma)$. If $v$ were a product vector,  this would again contradict block positivity of $\Gamma(\sigma).$ We thus have $vv^*\in \mathsf{PE}$ and use this to define the following map:
$$\Xi\colon \mathsf E\to \mathsf{PE};\    \sigma \mapsto vv^*.$$
In the following we will again consider the set $\mathsf M\subseteq\mathsf{P E}$ of maximally entangled pure states. Here it consists of those pure states that have the two identical  Schmidt coefficients $1/\sqrt{2}$ or, equivalently, have maximally mixed reduced density matrices  ($=\frac12 I_2$).

\begin{prop}
\label{Lem2}
    We have:
    \begin{enumerate}   
        \item[($i$)] $\Xi$ is continuous on $\mathsf E.$
        \item[($ii$)] $\Xi(\mathsf {PE})\subseteq \mathsf M$ and $\Xi^2_{|_\mathsf{M}}={\rm id}_{\mathsf M}$.  
        \item[($iii$)]  For all $\sigma \in \mathsf E,$ $t \in [0,1]$ we have $(1-t)\sigma+t \Xi^2(\sigma)\in \mathsf E.$
    \end{enumerate} 
\end{prop}

\begin{proof}
    ($i$) By functional calculus it is known that  the projector onto the eigenspace  of $\sigma$ corresponding to the eigenvalue $\lambda$  is $$\prod_{\mu\in {\rm spec}(\sigma)\setminus\{ \lambda\} }\frac{\sigma-\mu  I}{\lambda-\mu}.$$ 
    As by Weyl's inequality the spectrum of any Hermitian matrix is stable and by \Cref{unique} there is always a spectral gap for the smallest eigenvalue $\lambda$ on $\mathsf E$, continuity of $\Xi$ follows.
    
    ($ii$) Every pure entangled state admits a Schmidt-decomposition with Schmidt coefficients $s_1,s_2>0$ and  $s_1^2+s_2^2=1$. The choice of basis corresponds exactly to the local unitary action $({\rm SU}(2)\otimes {\rm SU}(2))\curvearrowright \mathsf {PE}$.  Thus we can write any such state as $$(U\otimes V) \begin{pmatrix}
    s_1^2 & 0 & 0 & s_1s_2\\
    0 & 0 & 0 & 0\\
    0 & 0 & 0 & 0\\
    s_1s_2 & 0 & 0 & s_2^2
    \end{pmatrix} (U \otimes V)^*$$ with $U,V\in {\rm SU}(2)$. We can now make use of the fact that local unitary action behaves well with regard to $\Gamma$, so application of $\Gamma$ yields: $$(U\otimes \overline V) \begin{pmatrix}
    s_1^2 & 0 & 0 & 0\\
    0 & 0 & s_1s_2 & 0\\
    0 & s_1s_2 & 0 & 0\\
    0 & 0 & 0 & s_2^2
    \end{pmatrix}(U\otimes \overline V)^*.$$ It is easy to see from this that the eigenvector corresponding to the unique negative eigenvalue $-s_1s_2$ is $$(U\otimes \overline V)\begin{psmallmatrix}
    0\\
    1\\
    -1\\
    0
    \end{psmallmatrix},$$ i.e.\ it lies in the orbits of the Bell states and is therefore maximally entangled.

    Now the maximally entangled pure states are exactly those with identical Schmidt coefficients $s_1=s_2=1/\sqrt{2}$.  The second claim then follows easily by the same calculations as above.  
    
    ($iii$) It remains to show that the segment from $\sigma$ to $\Xi^2(\sigma)$ stays entangled. Write $\Xi(\sigma)=vv^*$ and $\Xi^2(\sigma)=ww^*$. Then $w$ is the negative eigenvector of $\Gamma(vv^*)$. Since partial transpose is self-adjoint for the Hilbert-Schmidt pairing,
    $$
        \langle v,\Gamma(ww^*)v\rangle=\tr(vv^*\Gamma(ww^*))=\tr(\Gamma(vv^*)ww^*)=\langle w,\Gamma(vv^*)w\rangle<0.
        $$
    Also $\langle v,\Gamma(\sigma)v\rangle<0$ by definition of $v$. Hence $\langle v,\Gamma((1-t)\sigma+t\Xi^2(\sigma))v\rangle<0$ for all $t$, so the Peres--Horodecki criterion implies entanglement.
\end{proof}

With this preparation we get to the main result of this section:

\begin{theorem}
\label{Thm}
    In the case of two qubits, $\mathsf M$  is a strong deformation retract of $\mathsf E$. In particular, $\mathsf E$ is homotopy equivalent to $\mathsf M$, which is further homeomorphic to ${\rm SO}(3)$ and real projective space $\mathbb R\mathbb P^3$:$$\mathsf E\simeq \mathsf M\cong {\rm SO}(3)\cong \mathbb R\mathbb P^3.$$ In particular, one has $\pi_1(\mathsf E)\cong \Z/2\Z$ and the following integer homology groups  $$H_0(\mathsf E)=\Z,\ H_1(\mathsf E)=\Z/2\Z,\ H_3(\mathsf E)=\Z,$$ while all other homology groups of $\mathsf E$ are trivial.
\end{theorem}

\begin{proof}
    By \Cref{Lem2}, the map $\Xi^2\colon \mathsf E\to\mathsf M$ is continuous and the identity when restricted to $\mathsf M$. Furthermore, every  convex combination of $\sigma$ and $\Xi^2(\sigma)$ is entangled. Therefore $$H\colon\mathsf E\times [0,1] \to \mathsf E,\ (\sigma,t) \mapsto(1-t)\sigma + t\Xi^2(\sigma)$$ is a strong deformation retraction from $\mathsf E$ to $\mathsf M$.  
    So the homotopy and homology groups are those  stated in \Cref{prop:maxent}.
\end{proof}

\begin{remark}
    Another potential candidate, instead of $\Xi$, for inducing a homotopy equivalence between $\mathsf E$ and $\mathsf{PE}\simeq\mathsf M$ is the Lewenstein--Sanpera decomposition \cite{lsd}. In the two-qubit case, it was shown that for every state $\sigma$ there exists a unique state $L(\sigma)\in\mathsf{PE}$ such that
    \[
        \frac{1}{1-\lambda}\bigl(\sigma-\lambda L(\sigma)\bigr)\in\mathsf S,
    \]
    where $\lambda$ is chosen to be minimal. It is straightforward to see that $L$ is piecewise continuous when $\mathsf E$ is stratified according to rank, and it also seems plausible that $L$ is continuous on all of $\mathsf E$. However, upon closer inspection, a proof of this continuity turned out to be more elusive than we had anticipated, and we are currently unable to determine whether this approach can indeed be made to work.
\end{remark}

As it turns out, the topology of $\mathsf E$ changes substantially when passing to its closure. To prove this, we need some further preparation.

\begin{prop}
\label{Ner}
    For two qubits the equality $\mathsf S\cap \partial\Gamma(\mathsf D)= \partial_* \mathsf E$ holds.
\end{prop}
\begin{proof}
     The inclusion $\partial_*\mathsf E\subseteq \mathsf S\cap\partial\Gamma(\mathsf D)$ is clear  from the Peres-Horodecki criterion.
     For the other inclusion we have to show that for $\sigma\in\mathsf S\cap \partial\Gamma(\mathsf D)$ we always have  $\sigma\in\overline{\mathsf E}.$  

     If the kernel of $\Gamma(\sigma)$ is of dimension at least $2$, it contains a non-zero product vector $v$, as seen in the proof of \Cref{unique}. We can then replace $\sigma$ by $(1-\varepsilon)\sigma+\varepsilon \Gamma(vv^*)$ for arbitrary $\varepsilon>0,$ and have reduced to the case where the kernel of $\Gamma(\sigma)$ is of dimension $1$. 
     
     If the kernel of $\Gamma(\sigma)$ contains an entangled pure state $v$, we can find an entanglement witness for $v$, i.e.\ some $\rho\in\Gamma(\mathsf D)$ such that $v^*\rho v<0$. Then $(1-\varepsilon)\sigma+\varepsilon\Gamma(\rho)$ is an entangled state arbitrarily close to $\sigma$. 
   
    So the only remaining case is where the kernel of $\Gamma(\sigma)$ is spanned by a single product vector. By \cite{stv} we know that ${\rm rank}(\sigma)\geqslant 3$ must hold in this case. Now by assumption there are unit vectors $v,w\in \C^2$  such that $v\otimes\overline w$ spans the kernel of $\Gamma(\sigma)$. Then $$(v\otimes w)^*\sigma(v\otimes w)=(v\otimes \bar w)^*\Gamma(\sigma)(v\otimes \bar w)=0$$
    and so the kernel of $\sigma$ is spanned by $v\otimes w$.
    Choose unit vectors $a,b\in \C^2$ with $a^*v=0$ and $b^*w=0,$ and consider the traceless Hermitian matrix $$\phi=(v\otimes b)(a \otimes w)^*+(a \otimes w)(v\otimes b)^*,$$ for which we clearly have $v\otimes w\in {\rm kern}(\phi)={\rm range}(\phi)^{\perp}$. Since $\sigma$ is a strictly positive operator on the orthogonal complement of $v \otimes w$, we get $$\sigma_\varepsilon:=\sigma+\varepsilon \phi \geqslant 0$$ for $\varepsilon>0$ small enough.
    We have $$\Gamma (\sigma_\varepsilon)=\Gamma(\sigma)+\varepsilon(v\otimes \bar w)(a\otimes \bar b)^*+\varepsilon(a\otimes \bar b)(v \otimes \bar w)^*.$$
    Now define $x=\lambda(v\otimes \overline w)+(a\otimes \overline b)\in\C^2\otimes\C^2$ with $\lambda\in\RR$ and compute  \begin{align*}x^*\Gamma(\sigma_\varepsilon)x&=(a\otimes b)^*\sigma (a \otimes b)+2\varepsilon\lambda.\end{align*}
    For large enough negative $\lambda$, this clearly becomes negative, which shows that $\sigma_\varepsilon$ is entangled. This finishes the proof.
\end{proof}

\begin{theorem}\label{thm:closureE}
    For two qubits, $\overline{\mathsf E}$ and $\mathsf E$ are  dual in the sense that for all $i$:
    \begin{equation}\label{eq:closure-duality}
        \widetilde H^{\,13-i}(\overline{\mathsf E};\Z)\;\cong\;\widetilde H_i(\mathsf E;\Z).
    \end{equation}
    Moreover,
    $$
        H_0(\overline{\mathsf E})=\Z,\qquad
        H_{10}(\overline{\mathsf E})=\Z,\qquad
        H_{11}(\overline{\mathsf E})=\Z/2\Z,
    $$
    and all other homology groups of $\overline{\mathsf E}$ vanish.
\end{theorem}

\begin{proof}
    By \Cref{prop:Bor}, $\partial^*\mathsf E$ is a strong deformation retract of $\mathsf E$
    and $\partial_*\mathsf E$ of $\overline{\mathsf E}$, so
    $\widetilde H_i(\mathsf E)\cong\widetilde H_i(\partial^*\mathsf E)$ and
    $\widetilde H^{\,13-i}(\overline{\mathsf E})\cong\widetilde H^{\,13-i}(\partial_*\mathsf E)$
    for all $i$.

    The partial transpose $\Gamma$ is a linear involution of the ambient space $\RR^{15}$,
    hence a homeomorphism, and it satisfies $\Gamma(\mathsf S)=\mathsf S$ and
    $\Gamma\big(\partial\Gamma(\mathsf D)\big)=\partial\mathsf D$. By \Cref{Ner} it therefore
    restricts to a homeomorphism
    $$
        \partial_*\mathsf E=\mathsf S\cap\partial\Gamma(\mathsf D)
        \ \xrightarrow{\ \Gamma\ }\
        \mathsf S\cap\partial\mathsf D=:A .
    $$
    The set $A$ is nonempty, proper, compact and, being semialgebraic, locally contractible,
    and its complement in $\partial\mathsf D$ is $\partial^*\mathsf E$. Since $\mathsf D$ is a
    full-dimensional compact convex body in $\RR^{15}$, we have $\partial\mathsf D\cong\mathbb S^{14}$,
    and Alexander duality (\Cref{lem:alexander}) gives
    $$
        \widetilde H_i(\partial^*\mathsf E)\;\cong\;\widetilde H^{\,13-i}(A)
        \;\cong\;\widetilde H^{\,13-i}(\partial_*\mathsf E).
    $$
    Combining the three yields \eqref{eq:closure-duality}.

    By \Cref{Thm}, $\mathsf E\simeq\mathbb{RP}^3$, so $\widetilde H_i(\mathsf E)$ equals $\Z/2\Z$
    for $i=1$, $\Z$ for $i=3$, and $0$ otherwise. Hence, by \eqref{eq:closure-duality}, the 
    nontrivial reduced cohomology of $\overline{\mathsf E}$ is
    $$
        \widetilde H^{10}(\overline{\mathsf E})=\Z,\qquad
        \widetilde H^{12}(\overline{\mathsf E})=\Z/2\Z.
    $$
    Finally, $\overline{\mathsf E}$ is compact semialgebraic and thus homotopy equivalent to a
    finite simplicial complex \cite{real}*{Theorem 9.2.1}, so $H_*(\overline{\mathsf E};\Z)$ is finitely
    generated. The universal coefficient theorem \cite{hatcher}*{Corollary 3.3} then gives
    $\operatorname{rank}H_m=\operatorname{rank}H^m$ and
    $\operatorname{tor}H_m\cong\operatorname{tor}H^{\,m+1}$, which turns the cohomology above into
    the stated homology groups.
\end{proof}

\begin{remark}
    Comparing with \Cref{Thm}, passing to the closure moves the entire reduced topology of
    $\mathsf E$ from degrees $1$ and $3$ to degrees $11$ and $10$:
    $$
    \begin{array}{c|ccccc}
        k & 0 & 1 & 3 & 10 & 11\\\hline
        H_k(\mathsf E) & \Z & \Z/2\Z & \Z & 0 & 0\\[2pt]
        H_k(\overline{\mathsf E}) & \Z & 0 & 0 & \Z & \Z/2\Z
    \end{array}
    $$
\end{remark}

\begin{remark}
    $\partial_* \mathsf E$ and $\Gamma(\partial_* \mathsf E)$ are closed semialgebraic sets and by \Cref{Ner} we have $$\partial_*\mathsf  E\cup \Gamma(\partial_*\mathsf E)=\partial \mathsf S\cong \mathbb{S}^{14}.$$ Using the semialgebraic version of Mayer--Vietoris one gets: 
    $$H_i(\partial \mathsf D \cap \partial \Gamma(\mathsf D))\cong
    \begin{cases}
    \Z & \text{if } i=13\\
    \Z/2\Z \oplus \Z/ 2\Z &\text{if } i=11\\
    \Z \oplus \Z  & \text{if } i = 10\\
    \Z &\text{if } i=0\\ 0 & \text{else.}
    \end{cases}$$ 
\end{remark}

Finally, it turns out that the interior
has identical homology to $\mathsf E$:

\begin{theorem}
    $H_i(\mathsf E^\circ)\cong H_i(\mathsf E)$ for all $i$.
\end{theorem}

\begin{proof}
    We embed $\mathbb{R}^{15}$ into its one point compactification $\mathbb S^{15}$.     For the open sets $A=\mathsf D^\circ$ and $B=\Gamma(\mathsf D)^c$ we have  $$A\cup B\cong \Gamma(\mathsf D^\circ)\cup \mathsf D^c=(\Gamma(\mathsf D^\circ)^c \cap \mathsf D)^c=((\partial \Gamma(\mathsf D)\cap \mathsf D) \cup (\Gamma(\mathsf D)^c\cap \mathsf D))^c= \overline{\mathsf E}^c \ \mbox{ and }\  A \cap B =\mathsf E^\circ,$$ where for the last equality of the first calculation we again used \Cref{Ner} . We thus get the following reduced Mayer--Vietoris sequence: 
    $$\cdots \to \widetilde H_{i+1}(\overline{\mathsf E}^c)\to \widetilde H_i(\mathsf E^\circ)\to \widetilde H_i(\mathsf D^\circ)\oplus \widetilde H_i(\Gamma(\mathsf D)^c)\to \widetilde H_i(\overline{\mathsf E}^c)\to\cdots $$
    $\widetilde H_i(\mathsf D^\circ)$ and $\widetilde H_i(\Gamma (\mathsf D)^c)$ vanish in all dimensions (since we are working in $\mathbb S^{15}$), so again using Alexander duality (this time for one dimension higher), we get $$\widetilde{H}_i(\mathsf E^\circ) \cong \widetilde{H}_{i+1}(\overline{\mathsf E}^c) \cong \widetilde{H}^{13-i}(\overline{\mathsf E}).$$ With the above computed homology groups, this implies the result for all $i\geqslant 1.$ The result $H_0(\mathsf E^\circ)=\mathbb Z$ even holds in all dimensions $n_1,n_2$, as we have shown in \Cref{thm:path}.
\end{proof}

\end{document}